\documentclass[12pt]{article}

 \usepackage{amsmath,amssymb,amsthm} 
\usepackage{cite}
\newtheorem{prop}{Proposition} 
\newtheorem{lemm}{Lemma} 
\newtheorem{eg}{Example} 

\newcommand{\p}{\partial} 
\newcommand{\ta}{\tau} 

\begin{document}

\title{A new extended $q$-deformed KP hierarchy}
\author{ { Runliang Lin$^\dag$,
Xiaojun Liu$^\ddag$, Yunbo Zeng$^\dag$
\thanks{E-mail addresses: rlin@math.tsinghua.edu.cn (R.L.
Lin), lxj98@mails.tsinghua.edu.cn (X.J. Liu),
yzeng@math.tsinghua.edu.cn (Y.B. Zeng).
}} \\
 {\small $^\dag$ Department of Mathematical Sciences, Tsinghua
University, Beijing 100084, P.R. China}\\
{\small $^\ddag$ Department of Applied Mathematics,
  China Agricultural University,
  Beijing 100083, P.R. China  }
\date{}
}
 \maketitle




\begin{abstract}
A method is proposed in this paper to construct a new extended
$q$-deformed KP ($q$-KP) hiearchy and its Lax representation. This
new extended $q$-KP hierarchy contains two types of $q$-deformed KP
equation with self-consistent sources, and its two kinds of
reductions give the $q$-deformed Gelfand-Dickey hierarchy with
self-consistent sources and the constrained $q$-deformed KP
hierarchy, which include two types of $q$-deformed KdV equation with
sources and two types of $q$-deformed Boussinesq equation with
sources. All of these results reduce to the classical ones when $q$
goes to $1$. This provides a general way to construct (2+1)- and
(1+1)-dimensional $q$-deformed soliton equations with sources and
their Lax representations.

\end{abstract}

\vspace{2pc}

\maketitle

\section{Introduction}
In recent years, the $q$-deformed integrable systems attracted many
interests both in mathematics and in physics
\cite{KS,KC,Jimbo,Majid,Zhang,WZZ,FR,Frenkel,HL,AHvM,KLR,MS,Tu-1999a,Tu-1999b,Tu-2000,Iliev-1998a,Iliev-1998b,Iliev-2000,WWWW,HLC}.
The deformation is performed by using the $q$-derivative
$\partial_q$ to take the place of ordinary derivative $\partial_x$
in the classical systems, where $q$ is a parameter, and the
$q$-deformed integrable systems recover the classical ones as
$q\rightarrow 1$. The $q$-deformed $N$-th KdV ($q$-NKdV or
$q$-Gelfand-Dickey) hierarchy, the $q$-deformed KP ($q$-KP)
hierarchy, and the $q$-AKNS-D hierarchy were constructed, and some
of their integrable structures were also studied, such as the
infinite conservation laws, bi-Hamiltonian structure, tau function,
symmetries, B\"acklund transformation (see
\cite{Zhang,KLR,Tu-1999b,Tu-2000,WWWW,HLC} and the references
therein).

Multi-component generalization of an integrable model is a very
important subject
\cite{MR638807,MR723457,MR730247,MR688946,MR2006751,MR1621464,MR1629543,MR1107232,MR1738969}.
For example, the multi-component KP (mcKP) hierarchy given in
\cite{MR638807} contains many physically relevant nonlinear
integrable systems, such as Davey-Stewartson equation,
two-dimensional Toda lattice and three-wave resonant interaction
ones. Another type of coupled integrable systems is the soliton
equation with self-consistent sources, which has many physical
applications and can be obtained by coupling some suitable
differential equations to the original soliton equation
\cite{MR708435,MR910584,MR940618,LZM,ZML,Lin-2006,Hu-Wang,ZhangDJ}.
Very recently, we proposed a systematical procedure to construct a
new extended KP hierarchy and its Lax representation \cite{LZL}.
This new extended KP hierarchy contains two types of KP equation
with self-consistent sources (KPSCS-I and KPSCS-II), and its two
kinds of reductions give the Gelfand-Dickey hierarchy with
self-consistent sources \cite{MR1165512} and the $k$-constrained KP
 hierarchy \cite{MR1117170,MR1185854}.
In fact, the approach
which we proposed in \cite{LZL} in the framework of Sato theory
        can be applied to construct many other extended (2+1)-dimensional soliton hierarchies,
        such as BKP hierarchy, CKP hierarchy and DKP hierarchy,
        and provides a general way to obtain (2+1)-dimensional and (1+1)-dimensional
        integrable soliton hierarchies with self-consistent sources.

The KdV equation with self-consistent sources and the KP equation
with self-consistent sources
    can describe the interaction of long and short waves
    (see \cite{MR708435,MR910584,MR940618,LZM,ZML,Lin-2006,Hu-Wang,ZhangDJ} and the references therein).
    In contrast with the well-studied KdV and KP equation with self-consistent sources,
    the $q$-Gelfand-Dickey hierarchy with self-consistent sources
    and the $q$-KP hierarchy with self-consistent sources have not been investigated yet.
    It is ineresting to consider the case of the algebra of $q$-pseudo-differential operator,
    and to see if our approach could be generalized to construct new extended $q$-deformed integrable systems,
    which would enable us to find two types of new $q$-deformed soliton equation with sources in a systematic way.
    In this paper, we will give a systematical procedure to construct a
new extended $q$-deformed KP ($q$-KP) hierarchy and its Lax
representation. First, we define a new vector filed
$\partial_{\tau_k}$ by a linear combination of all vector fields
$\partial_{t_n}$ in ordinary $q$-deformed KP hierarchy, then we
introduce a new Lax type equation which consists of the
$\tau_k$-flow and the evolutions of wave functions. Under the
evolutions of wave functions, the commutativity of
$\partial_{\tau_k}$-flow and $\partial_{t_n}$-flows gives rise to a
new extended $q$-KP hierarchy. This new extended $q$-KP hierarchy
contains two types of $q$-deformed KP equation with self-consistent
sources ($q$-KPSCS-I and $q$-KPSCS-II), and its two kinds of
reductions give the $q$-deformed Gelfand-Dickey hierarchy with
self-consistent sources and the constrained $q$-deformed KP
hierarchy, which are some $(1+1)$-dimensional $q$-deformed soliton
equation with self-consistent sources, e.g., two types of
$q$-deformed KdV equation with self-consistent sources ($q$-KdVSCS-I
and $q$-KdVSCS-II) and two types of $q$-deformed Boussinesq equation
with self-consistent sources ($q$-BESCS-I and $q$-BESCS-II). The
$q$-KdVSCS-II is just the $q$-deformed Yajima-Oikawa equation. All
of these results reduce to the classical ones when $q \rightarrow
1$. Thus, the method proposed in this paper is a general way to find
the $(1+1)$- and $(2+1)$-dimensional  $q$-deformed soliton equation
with self-consistent sources and their Lax representations. It
should be noticed that a general setting of ``pseudo-differential"
operators
        on regular time scales has been proposed to construct some integrable systems \cite{GGS2005,BSS2008},
        where the $q$-differential operator is just a particular case.
        Our paper will be organized as follows. In section 2, we will recall
some notations in the $q$-calculus and construct the new extended
$q$-KP hierarchy, and then two types of $q$-deformed KP equation
with sources will be presented. In section 3, the two kinds of
reductions for the new extended $q$-KP hierarchy will be considered,
and some $(1+1)$-dimensional $q$-deformed soliton equation with
self-consistent sources will be deduced. In section 4, some
conclusions will be given.

\section{New extended $q$-deformed KP hierarchy}

In this section, we will give a procedure to construct a new
extended $q$-KP hierarchy and its Lax representation. Then, as the
examples, two types of $q$-deformed KP equation with self-consistent
sources ($q$-KPSCS-I and $q$-KPSCS-II) will be presented explicitly.

The $q$-deformed differential operator $\partial_q$ is defined as
    $$ \partial_q(f(x))=\frac{f(qx)-f(x)}{x(q-1)}, $$
which recovers the ordinary differentiation $\partial_x(f(x))$ as
$q\rightarrow 1$. Let us define the $q$-shift operator $\theta$ as
    $$ \theta(f(x))=f(qx). $$
Then we have the $q$-deformed Leibnitz rule
    $$ \partial_q^n f=\sum_{k\ge0}\left(\begin{array}{c}{n}\\{k}\end{array}\right)_q\theta^{n-k}(\partial_q^kf)\partial_q^{n-k},\qquad n\in
     {\mathbb  Z}, $$
where the $q$-number and the $q$-binomial are defined by
    $$ (n)_q=\frac{q^n-1}{q-1},
\qquad
   \left(\begin{array}{c}{n}\\{k}\end{array}\right)_q=\frac{(n)_q(n-1)_q\cdots(n-k+1)_q}{(1)_q(2)_q\cdots(k)_q},\qquad
   \left(\begin{array}{c}{n}\\{0}\end{array}\right)_q=1. $$

For a $q$-pseudo-differential operator ($q$-PDO) of the form
    $$P=\sum\limits_{i=-\infty}^np_i\partial_q^i, $$
we decompose $P$ into the differential part  and the integral part
as follows
    $$ P_+=\sum\limits_{i\ge0}p_i\partial_q^i, \qquad
    P_-=\sum\limits_{i\leq-1}p_i\partial_q^i.$$
The conjugate operation ``$*$'' for $P$ is defined by
    $$ P^*=\sum\limits_i(\partial_q^*)^ip_i, \qquad
    \partial_q^*=-\partial_q\theta^{-1}=-\frac{1}{q}\partial_{\frac{1}{q}}.$$

The $q$-KP hierarchy is defined by the Lax equation (see, e.g.,
\cite{Iliev-1998b})
\begin{equation}
\label{qKP-Lax}
    \partial_{t_n} L = [B_n,L], \qquad B_n=L^n_+,
\end{equation}
with Lax operator of the form
\begin{equation}
    L=\partial_q+\sum_{i=0}^{\infty} u_i \partial_q^{-i}.
\end{equation}
According to the Sato theory, we can express the Lax operator as a
dressed operator
\begin{equation}
    L=S\partial_q S^{-1},
\end{equation}
where $S=1+\sum\limits_{i=1}^\infty S_i\partial_q^{-i}$ is called
the Sato operator and $S^{-1}$ is its formal inverse. The Lax
equation (\ref{qKP-Lax}) is equivalent to the Sato equation
\begin{equation}
\label{eqn:ev-S-q}
    S_{t_n}=-(L^n)_- S.
\end{equation}

The $q$-wave function $w_q(x,\overline t;z)$ and $q$-adjoint wave
function $w^*(x,\overline t;z)$ (here $\overline{t}=(t_1,t_2,t_3,
\ldots)$) are defined as follows
\begin{subequations}
\label{eqns:ev-wave-q}
\begin{equation}
    w_q=S
    e_q(xz)\exp\left(\sum\limits_{i=1}^{\infty}t_iz^i\right),
\end{equation}
\begin{equation}
    w^*=(S^*)^{-1}|_{x/q}
    e_{1/q}(-xz)\exp\left(-\sum\limits_{i=1}^{\infty}t_iz^i\right),
\end{equation}
\end{subequations}
where the notation $P|_{x/t}=\sum\limits_i
p_i(x/t) t^i
\partial_q^i $ (for $P=\sum\limits_i p_i(x)\partial_q^i$) is used,
and
    $$ e_q(x)=\exp\left(\sum_{k=1}^{\infty}\frac{(1-q)^k}{k(1-q^k)}x^k\right). $$
It is easy to show that $w_q$ and $w_q^*$ satisfy the following
linear systems
    $$ Lw_q=z w_q, \qquad \frac{\partial w_q}{\partial t_n}=B_n w_q,  $$
    $$ L^*|_{x/q}w^*_q=zw^*_q, \qquad
    \frac{\partial w^*_q}{\partial t_n}=-(B_n|_{x/q})^*w^*_q. $$

It can be proved that \cite{Tu-1999b}
\begin{equation}
\label{eqns:res-lema}
  T(z)_-\equiv\sum_{i\in{\mathbb  Z}} L^i_ -z^{-i-1}=w_q \partial_q^{-1} \theta(w_q^*).
\end{equation}

For any fixed $k\in{\mathbb  N}$, we define a new variable $\tau_k$
whose vector field is
    $$\partial_{\tau_k}=\partial_{t_k}-\sum_{i=1}^N\sum_{s\ge0}\zeta_i^{-s-1}\partial_{t_s},$$
where $\zeta_i$'s are arbitrary distinct non-zero parameters. The
$\tau_k$-flow is given by
\begin{eqnarray*}
  L_{\tau_k}=\partial_{t_k}L-\sum_{i=1}^N\sum_{s\ge 0}\zeta_i^{-s-1}\partial_{t_s}L
  =[B_k,L]-\sum_{i=1}^N\sum_{s\ge0}\zeta_i^{-s-1}[B_s,L]\\
  =[B_k,L]+\sum_{i=1}^N\sum_{s\in{\mathbb  N}}\zeta_i^{-s-1}[L^s_-,L]
  =[B_k,L]+\sum_{i=1}^N\sum_{s\in{\mathbb  Z}}\zeta_i^{-s-1}[L^s_-,L].
\end{eqnarray*}
Define $\tilde{B}_k$ by
\begin{displaymath}
  \tilde{B}_k=B_k+\sum_{i=1}^N\sum_{s\in{\mathbb  Z}}\zeta_i^{-s-1}L^s_-,
\end{displaymath}
which, according to (\ref{eqns:res-lema}), can be written as
\begin{displaymath}
  \tilde{B}_k=B_k+\sum_{i=1}^N w_q(x,\overline t;\zeta_i)
  \partial_q ^{-1} \theta(w_q^*(x,\overline t;\zeta_i)).
\end{displaymath}
By setting $\phi_i=w_q(x,\overline t;\zeta_i)$,
$\psi_i=\theta(w_q^*(x,\overline t;\zeta_i))$, we have
\begin{subequations}
\begin{equation}
  \tilde{B}_k=B_k+\sum_{i=1}^N \phi_i\partial_q^{-1} \psi_i,
\end{equation}
where $\phi_i$ and $\psi_i$ satisfy the following equations
\begin{equation}
  \phi_{i,t_n}=B_n(\phi_i),\qquad \psi_{i,t_n}=-B^*_n(\psi_i), \qquad i=1,\cdots,N.
\end{equation}
\end{subequations}
Now we introduce a new Lax type equation given by
\begin{subequations}
\label{eqns:qKP-nLax}
\begin{equation}
\label{eqns:qKP-nLax-LBn}
  L_{\tau_k}=[B_k+\sum_{i=1}^N \phi_i\partial_q^{-1}\psi_i,L].
\end{equation}
with
\begin{equation}
\label{eqns:q-r-qKP}
  \phi_{i,t_n}=B_n(\phi_i),\qquad \psi_{i,t_n}=-B_n^*(\psi_i),\qquad i=1,\cdots,N.
\end{equation}
\end{subequations}
We have the following lemma
\begin{lemm}
  \label{lemm} $[B_n,\phi\partial_q^{-1}\psi]_-=B_n(\phi)\partial_q^{-1}\psi-\phi\partial_q^{-1}B_n^*(\psi)$.
\end{lemm}
\begin{proof}
  Without loss of generality, we consider a monomial: $P=a\partial_q^n$ ($n\ge 1$). Then
  \begin{displaymath}
    [P,\phi\partial_q^{-1}\psi]_-=a (\partial_q^n(\phi))\partial_q^{-1}\psi- (\phi\partial_q^{-1}\psi a\partial_q^n)_-.
  \end{displaymath}
  Notice that the second term can be rewritten in the following way
  \begin{eqnarray*}
    (\phi\partial_q^{-1}\psi a\partial_q^{n})_-
    =\phi(\theta^{-1}(\psi a))\partial_q^{n-1}
        -\phi\partial_q^{-1}(\partial_q\theta^{-1}(a\psi))\partial_q^{n-1})_- \\
    =(\phi\partial_q^{-1}(-\partial_q\theta^{-1}(a\psi))\partial_q^{n-1})_-
    =\cdots
    =\phi\partial_q^{-1}\left((-\partial_q\theta^{-1})^n
    (a\psi)\right)=\phi\partial_q^{-1}P^*(\psi),
  \end{eqnarray*}
  then the lemma is proved.
\end{proof}

\begin{prop}
  (\ref{qKP-Lax}) and (\ref{eqns:qKP-nLax}) give rise to the following new extended $q$-deformed KP hierarchy
\begin{subequations}
    \label{eqns:nmcqKP}
    \begin{eqnarray}
      &B_{n,\tau_k}-(B_k+\sum_{i=1}^N\phi_i\partial_q^{-1}\psi_i)_{t_n}
      +[B_n,B_k+\sum_{i=1}^N\phi_i\partial_q^{-1}\psi_i]=0\label{eqn:nmcqKP-zc}\\
      &\phi_{i,t_n}=B_n(\phi_i),\label{eqn:nmcqKP-q}\\
      &\psi_{i,t_n}=-B_n^*(\psi_i),\quad i=1,\cdots,N.\label{eqn:nmcqKP-r}
    \end{eqnarray}
 \end{subequations}
\end{prop}
\begin{proof}
  We will show that under (\ref{eqns:q-r-qKP}), (\ref{qKP-Lax}) and (\ref{eqns:qKP-nLax-LBn}) give rise to
  (\ref{eqn:nmcqKP-zc}). For convenience, we assume $N=1$, and denote $\phi_1$ and $\psi_1$ by $\phi$ and $\psi$,
  respectively.
  By  (\ref{qKP-Lax}), (\ref{eqns:qKP-nLax}) and Lemma~\ref{lemm}, we have
  \begin{eqnarray*}
    B_{n,\tau_k}=(L^n_{\tau_k})_+=[B_k+\phi\partial_q^{-1}\psi,L^n]_+
    \\
    =[B_k+\phi\partial_q^{-1}\psi,L^n_+]_+ +[B_k+\phi\partial_q^{-1}\psi,L^n_-]_+\\
    =[B_k+\phi\partial_q^{-1}\psi,L^n_+]-[B_k+\phi\partial_q^{-1}\psi,L^n_+]_-+[B_k,L^n_-]_+\\
    =[B_k+\phi\partial_q^{-1}\psi,B_n]-[\phi\partial_q^{-1}\psi,B_n]_-+[B_n,L^k]_+\\
    =[B_k+\phi\partial_q^{-1}\psi,B_n]+B_n(\phi)\partial_q^{-1}\psi-\phi\partial_q^{-1}B_n^*(\psi)+B_{k,t_n}\\
    =[B_k+\phi\partial_q^{-1}\psi,B_n]+(B_k+\phi\partial_q^{-1}\psi)_{t_n}.
  \end{eqnarray*}
\end{proof}
Under (\ref{eqn:nmcqKP-q}) and (\ref{eqn:nmcqKP-r}), the Lax
representation for (\ref{eqn:nmcqKP-zc}) is given by
\begin{subequations}
\begin{eqnarray}
  \Psi_{\tau_k}=(B_k+\sum_{i=1}^N\phi_i\partial_q^{-1}\psi_i)(\Psi), \label{eqn:qKP-Lax-1}\\
  \Psi_{t_n}=B_n(\Psi). \label{eqn:qKP-Lax-2}
\end{eqnarray}
\end{subequations}

\parindent=0pt
{\bf Remark 1} The main step in our approach is to define a new Lax
equation (\ref{eqns:qKP-nLax}).
        For the extended KP hierarchy in \cite{LZL}, a similar formula like (\ref{eqns:qKP-nLax})
        can be motivated by the well-known $k$-constraint of KP hierarchy,
        which is obtained by imposing
        $ L^k=B_k+\sum\limits_{i=1}^N \phi_i\partial^{-1}\psi_i$.
        Here, the formula (\ref{eqns:qKP-nLax}) can also be motivated by
        the $k$-constraint of $q$-KP hierarchy as given in \cite{Tu-1999b}.
        This enables us to obtain the $k$-constrained $q$-KP hierarchy
        and the $q$-Gelfand-Dickey hierarchy with sources by dropping the $\tau_k$-dependence
        and $t_n$-dependence in the new extended $q$-KP hierarchy (\ref{eqns:nmcqKP}) respectively (see Section 3).

\parindent=0pt
{\bf Remark 2} When taking $\phi_i=\psi_i=0$, $i=1,\ldots,N$,
    then the extended $q$-KP hierarchy (\ref{eqns:nmcqKP}) reduces to the $q$-KP hierarchy.

\parindent=0pt
{\bf Remark 3} Integrable systems can be constructed
        from the algebra of ``pseudo-differential" operators on regular time scales in \cite{GGS2005,BSS2008},
        where the algebra of $q$-``pseudo-differential" operator is a particular case.
        In fact, our approach for constructing new extended integrable systems
        can also be generalized to the general setting as in \cite{GGS2005,BSS2008}.

For convenience, we write out some operators here
\begin{eqnarray*}
B_1=\p_q+ u_0, \qquad 
B_2=\p^2_q + v_1\p_q+ v_0,\qquad 
B_3=\p^3_q + {s}_2\p^2_q+ {s}_1\p_q+ {s}_0,
\end{eqnarray*}
\begin{eqnarray*}
\phi_i\partial_q^{-1}\psi_i =
r_{i1}\partial_q^{-1}+r_{i2}\partial_q^{-2}+r_{i3}\partial_q^{-3}+\ldots,
\qquad i=1,\ldots,N,
\end{eqnarray*}
where
\begin{eqnarray*}
v_1=\theta(u_0)+u_0, \qquad 
v_0= (\p_q u_0)+\theta(u_{1})+u_0^2+u_{1},\\ 
v_{-1}=(\p_q u_{1})+\theta(u_{2})+u_0u_{1}+u_{1}\theta^{-1}(u_0)+u_{2},
\end{eqnarray*}
\begin{eqnarray*}
{s}_2=\theta(v_1)+ u_0, \qquad 
 {s}_1=(\p_q v_1) + \theta(v_0)+ u_0v_1+u_{1}, \\ 
 {s}_0=(\p_q v_0) + \theta(v_{-1}) + u_0v_0+u_{1}\theta^{-1}(v_{1})+u_{2}. 
\end{eqnarray*}
\begin{eqnarray*}
 {r}_{i1}=\phi_i\theta^{-1}(\psi_i), \qquad 
 {r}_{i2}=-\frac{1}{q} \phi_i \theta^{-2}(\partial_q\psi_i), \qquad 
 {r}_{i3}=\frac{1}{q^3}\phi_i\theta^{-3}(\partial_q^2\psi_i). 
\end{eqnarray*}
and $v_{-1}$ comes from $L^2=B_2+ v_{-1}\p^{-1}_q+
v_{-2}\p^{-2}_q+\cdots$.

Then, one can compute the following commutators
    $$ [B_2,B_3]=f_2\partial_q^2+f_1\partial_q+f_0, \qquad
        [B_2,\phi_i\partial_q^{-1}\psi_i]=g_{i1}\partial_q+g_{i0}+\ldots,$$
    $$ [B_3,\phi_i\partial_q^{-1}\psi_i]=h_{i2}\partial_q^2+h_{i1}\partial_q+h_{i0}+\ldots,\qquad i=1,\ldots,N,$$
where
\begin{eqnarray*}
 f_2=\partial_q^2 s_2 
+(q+1)\theta(\partial_q s_1) +\theta^2(s_0) 
+v_1 \partial_q s_2 
+v_1 \theta(s_1) 
+v_0 s_2 
-(q^2+q+1) \theta(\partial_q^2 v_1) \\ 
-(q^2+q+1) \theta^2(\partial_q v_0 ) 
-(q+1) s_2 \theta(\partial_q v_1 ) 
-s_2 \theta^2(v_0) 
-s_1 \theta(v_1)-s_0, 
\end{eqnarray*}
\begin{eqnarray*}
 f_1=\partial_q^2 s_1 +(q+1)\theta(\partial_q s_0 ) 
+v_1\partial_q s_1 
+v_1\theta(s_0) 
+v_0s_1 -\partial_q^3 v_1 
-(q^2+q+1)\theta( \partial_q^2 v_0 ) \\
-s_2\partial_q^2 v_1  
-(q+1)s_2\theta( \partial_q v_0 ) 
-s_1\partial_q v_1  
-s_1\theta( v_0) 
-s_0v_1, 
\end{eqnarray*}
\begin{eqnarray*}
 f_0=\partial_q^2  s_0  +v_1 \partial_q s_0  
-\partial_q^3 v_0  -s_2 \partial_q^2 v_0  
-s_1 \partial_q v_0, 
\end{eqnarray*}
\begin{eqnarray*}
 g_{i1}=\theta^2(r_{i1}) -r_{i1},
\qquad
     g_{i0}=(q+1) \theta( \partial_q r_{i1} )
+\theta^2(r_{i2}) 
+v_1 \theta( r_{i1}) 
-r_{i1} \theta^{-1}(v_1) 
-r_{i2}, 
\end{eqnarray*}
\begin{eqnarray*}
 h_{i2}= \theta^3(r_{i1})-r_{i1},  
\qquad
     h_{i1}=(q^2+q+1) \theta^2(\partial_q r_{i1} )  
+\theta^3(r_{i2})+s_2 \theta^2(r_{i1}) 
-r_{i1} \theta^{-1}(s_2). 
\end{eqnarray*}
\begin{eqnarray*}
 h_{i0}=(q^2+q+1) \theta( \partial_q^2 r_{i1} ) 
+(q^2+q+1) \theta^2(\partial_q r_{i2} )
+\theta^3(r_{i3}) 
+(q+1) s_2  \theta( \partial_q r_{i1} )  \\ 
+s_2  \theta^2(r_{i2}) 
+s_1  \theta( r_{i1}) 
-r_{i1} \theta^{-1}( s_1 ) 
+\frac{1}{q} r_{i1} \theta^{-2}(\partial_q s_2 ) 
-r_{i2} \theta^{-2}(s_2 ) 
-r_{i3}. 
\end{eqnarray*}

 Now, we list some examples in the new extended $q$-KP
hierarchy (\ref{eqns:nmcqKP}).

\begin{eg} [The first type of $q$-KPSCS ($q$-KPSCS-I)]
  For $n=2$ and $k=3$, (\ref{eqns:nmcqKP}) yields the first type of
  $q$-deformed KP equation with self-consistent sources ($q$-KPSCS-I) as
  follows
 \begin{subequations}
   \label{eqns:qKPSCS-1}
    \begin{eqnarray}
      & -\frac{\partial s_2}{\partial t_2}+f_2=0, \\ 
      & \frac{\partial v_1}{\partial \tau_3}-\frac{\partial s_1}{\partial t_2}+f_1+\sum_{i=1}^N g_{i1}=0, \\
      & \frac{\partial v_0}{\partial \tau_3}-\frac{\partial s_0}{\partial t_2}+f_0+\sum_{i=1}^N g_{i0}=0, \\
      & \phi_{i,t_2} = B_2(\phi_i), \qquad
      \psi_{i,t_2}=-B_2^*(\psi_i), \qquad i=1,\ldots,N.
    \end{eqnarray}
\end{subequations}
The Lax representation for (\ref{eqns:qKPSCS-1}) is
 \begin{subequations}
\begin{eqnarray}
  \Psi_{\tau_3}=(\partial_q^3+s_2\partial_q^2+s_1\partial_q+s_0+\sum_{i=1}^N\phi_i\partial_q^{-1}\psi_i)(\Psi),\\
  \Psi_{t_2}=(\partial_q^2+v_1\partial_q+v_0)(\Psi).
\end{eqnarray}
 \end{subequations}

Let $q\rightarrow 1$ and $u_0\equiv 0$, then the $q$-KPSCS-I
(\ref{eqns:qKPSCS-1}) reduces to the first type of KP equation with
self-consistent sources (KPSCS-I) which reads as
\cite{MR708435,MR910584}
 \begin{subequations}
    \begin{eqnarray}
       u_{1,t_2}-u_{1,xx}-2u_{2,x}=0, \\ 
      2 u_{1,\tau_3}-3 u_{2,t_2}-3 u_{1,x,t_2}
    + u_{1,xxx}+3 u_{2,xx}-6 u_1 u_{1,x}
    +2 \partial_x\sum_{i=1}^N  \phi_i \psi_i=0,\\
      \phi_{i ,t_2}-\phi_{i ,xx}-2 u_1  \phi_i=0, \\
      \psi_{i ,t_2}+ \psi_{i ,xx}+2 u_1  \psi_i=0, \qquad
i=1,\ldots,N.
    \end{eqnarray}
 \end{subequations}
\end{eg}

\begin{eg}[The second type of $q$-deformed KPSCS ($q$-KPSCS-II)]
  For $n=3$ and $k=2$, (\ref{eqns:nmcqKP}) yields the second type of
  $q$-deformed KP equation with self-consistent sources ($q$-KPSCS-II) as
  follows
  \begin{subequations}
  \label{eqns:qKPSCS-2}
    \begin{eqnarray}
      & \frac{\partial s_2}{\partial \tau_2}-f_2+\sum_{i=1}^N h_{i2}=0,\\
      & \frac{\partial s_1}{\partial \tau_2}-\frac{\partial v_1}{\partial t_3}-f_1+\sum_{i=1}^N h_{i1}=0,\\
      & \frac{\partial s_0}{\partial \tau_2}-\frac{\partial v_0}{\partial t_3}-f_0+\sum_{i=1}^N h_{i0}=0, \\
      & \phi_{i,t_3} = B_3(\phi_i), \qquad
      \psi_{i,t_3}=-B_3^*(\psi_i), \qquad i=1,\ldots,N.
    \end{eqnarray}
 \end{subequations}
The Lax representation for (\ref{eqns:qKPSCS-2}) is
\begin{subequations}
\begin{eqnarray}
  \Psi_{\tau_2}=(\partial_q^2+v_1\partial_q+v_0+\sum_{i=1}^N\phi_i\partial_q^{-1}\psi_i)(\Psi),\\
  \Psi_{t_3}=(\partial_q^3+s_2\partial_q^2+s_1\partial_q+s_0)(\Psi).
\end{eqnarray}
\end{subequations}
Let $q\rightarrow 1$ and $u_0\equiv 0$, then the $q$-KPSCS-II
(\ref{eqns:qKPSCS-2}) reduces to the second type of KP equation with
self-consistent sources (KPSCS-II) which reads as \cite{MR708435}
  \begin{subequations}
    \begin{eqnarray}
       u_{1,\tau_2}-u_{1,xx}-2 u_{2,x}  + \partial_x\sum_{i=1}^N \phi_i \psi_i=0, \\ 
       3  u_{2,\tau_2}+3  u_{1,x,\tau_2}-2  u_{1,t_3}
      - u_{1,xxx}+6 u_1  u_{1,x}-3  u_{2,xx}
    +3 \partial_x\sum_{i=1}^N  \phi_{i,x} \psi_i=0, \\ 
       \phi_{i,t_3}-\phi_{i,xxx}
     -3 u_1  \phi_{i,x}
     -3 (u_{1,x}+u_2 ) \phi_i=0, \\ 
      \psi_{i ,t_3}-\psi_{i ,xxx}
     -3 u_1  \psi_{i ,x}+3 u_2  \psi_i=0, \qquad i=1,\ldots,N.
    \end{eqnarray}
  \end{subequations}
\end{eg}

\section{Reductions}
\label{sec:reductions} The new extended $q$-deformed KP hierarchy
(\ref{eqns:nmcqKP}) admits reductions to several well-known
$q$-deformed $(1+1)$-dimensional systems.

\subsection{The $n$-reduction of (\ref{eqns:nmcqKP})}
The $n$-reduction is given by
\begin{equation}
  \label{eqn:n-redu-q}
  L^n=B_n \qquad \mbox{or}\qquad L^n_-=0,
\end{equation}
then (\ref{eqns:ev-wave-q}) implies that
 \begin{subequations}
  \label{eqns:ev-wave-redu-q}
  \begin{eqnarray}
    &B_n(\phi_i)=L^n \phi_i=\zeta_i^n \phi_i,\\
    &-B_n^*(\psi_i)=-L^{n*}\psi_i=-\zeta_i^n\psi_i.
  \end{eqnarray}
\end{subequations}
By using Lemma 1 and (\ref{eqns:ev-wave-redu-q}), we can see that
the constraint (\ref{eqn:n-redu-q}) is invariant under the $\tau_k$
flow
\begin{eqnarray}
  (L^n_-)_{\tau_k}=[B_k,L^n]_-+\sum_{i=1}^N[\phi_i\partial_q^{-1}\psi_i,L^n]_-\nonumber\\
  =[B_k,L^n_-]_-+\sum_{i=1}^N[\phi_i\partial_q^{-1}\psi_i,L^n_+]_-+\sum_{i=1}^N[\phi_i\partial_q^{-1}\psi_i,L^n_-]_-\nonumber\\
  =\sum_{i=1}^N [\phi_i\partial_q^{-1}\psi_i,B_n]_-
  =-\sum_{i=1}^N (\phi_{i,t_n}\partial_q^{-1}\psi_i+\phi_i\partial_q^{-1}\psi_{i,t_n})\nonumber\\
  =-\sum_{i=1}^N(\zeta_i^n\phi_i\partial_q^{-1}\psi_i-\zeta_i^n\phi_i\partial_q^{-1}\psi_i)=0.\label{eqn:n-invar-q}
\end{eqnarray}
The equations (\ref{eqn:n-redu-q}) and (\ref{eqn:ev-S-q}) imply that
$S_{t_n}=0$, so $(L^k)_{t_n}=0$, which together with
(\ref{eqn:n-invar-q}) means that one can drop $t_n$ dependency from
(\ref{eqns:nmcqKP}) and obtain
 \begin{subequations}
  \label{eqs:GDHSCS-q}
  \begin{eqnarray}
    B_{n,\ta_k}=[(B_n)^{\frac k n}_++\sum_{i=1}^N \phi_i\partial_q^{-1}\psi_i,B_n],\\
    B_n(\phi_i)=\zeta_i^n\phi_i,\\
    B_n^*(\psi_i)=\zeta_i^n\psi_i,\quad i=1,\cdots,N.
  \end{eqnarray}
 \end{subequations}
The system (\ref{eqs:GDHSCS-q}) is the $q$-deformed Gelfand-Dickey
hierarchy with self-consistent sources.

\begin{eg}[The firs tyep of $q$-deformed KdVSCS ($q$-KdVSCS-I)]
For $n=2$ and $k=3$, (\ref{eqs:GDHSCS-q}) presents the first type of
$q$-deformed KdV equation with self-consistent sources
($q$-KdVSCS-I)
 \begin{subequations}
\label{eqns:qKdVSCS-I}
\begin{eqnarray}
  & v_{1,\tau_3}+f_1+\sum_{i=1}^N g_{i1}=0 , \\
  & v_{0,\tau_3}+f_0+\sum_{i=1}^N g_{i0}=0, \\
  & u_2+\theta(u_2)+\partial_q(u_1)+u_0 u_1+u_1 \theta^{-1}(u_0)=0, \\
  & (\partial_q^2 +v_1 \partial_q  +v_0)  (\phi_i)-\zeta^2 \phi_i=0, \\
  & (\partial_q^2 +v_1 \partial_q  +v_0)^* (\psi_i) -\zeta^2 \psi_i=0,\qquad i=1,\cdots,N,
\end{eqnarray}
\end{subequations}
with the Lax representation
\begin{eqnarray*}
\Psi_{\tau_3}=(\partial_q^3+s_2\partial_q^2+s_1\partial_q+s_0+\sum_{i=1}^N
  \phi_i\partial_q^{-1}\psi_i)(\Psi),\\
 (\partial_q^2+v_1\partial_q+v_0)(\Psi)=\lambda\Psi,\qquad
  u_2+\theta(u_2)+\partial_q(u_1)+u_0 u_1+u_1 \theta^{-1}(u_0)=0.
\end{eqnarray*}
Let $q\rightarrow 1$ and $u_0\equiv 0$, then the $q$-KdVSCS-I
(\ref{eqns:qKdVSCS-I}) reduces to the first type of KdV equation
with self-consistent sources (KdVSCS-I) which reads as
\begin{eqnarray*}
  & u_2=-\frac{1}{2} u_{1,x},\\
  &   u_{1,\tau_3}-3 u_1  u_{1,x}
    - \frac {1}{4}  u_{1,xxx}+  \partial_x\sum_{i=1}^N \phi_i \psi_i=0,\\
  & \phi_{i ,xx}+2 u_1  \phi_i -\zeta^2 \phi_i=0,\\
  & \psi_{i ,xx}+2 u_1  \psi_i -\zeta^2 \psi_i=0,\qquad i=1,\cdots,N. 
\end{eqnarray*}
The first type of KdV equation with self-consistent sources
(KdVSCS-I) can be solved by the inverse scattering method
\cite{MR940618,LZM} or by the Darboux transformation (see
\cite{Lin-2006} and the references therein).

\end{eg}

\begin{eg} [The first type of $q$-BESCS ($q$-BESCS-I)]

For $n=3$ and $k=2$, (\ref{eqs:GDHSCS-q}) presents the first type of
$q$-deformed Boussinesq equation with self-consistent sources
($q$-BESCS-I)
 \begin{subequations}
\label{eqns:qBESCS-I}
\begin{eqnarray}
  & s_{2,\tau_2}-f_2+\sum_{i=1}^N h_{i2}=0, \\
  & s_{1,\tau_2}-f_1+\sum_{i=1}^N h_{i1}=0, \\
  & s_{0,\tau_2}-f_0 +\sum_{i=1}^N h_{i0} =0, \\
  & (\partial_q^3 +s_2 \partial_q^2+s_1 \partial_q +s_0) (\phi_i)-\zeta^3 \phi_i=0,\\
  & (\partial_q^3 +s_2 \partial_q^2+s_1 \partial_q +s_0)^* (\psi_i)-\zeta^3 \psi_i=0,\qquad   i=1,\cdots,N,
\end{eqnarray}
\end{subequations}
with the Lax representation
\begin{displaymath}
 \Psi_{\tau_2}=(\partial_q^2+v_1\partial_q+v_0+\sum_{i=1}^N
\phi_i\partial_q^{-1}\psi_i)(\Psi), \qquad
  (\partial_q^3+s_2\partial_q^2+s_1\partial_q+s_0)(\Psi)=\lambda\Psi.
\end{displaymath}
Let $q\rightarrow 1$ and $u_0\equiv 0$, then the $q$-BESCS-I
(\ref{eqns:qBESCS-I}) reduces to the first type of Boussinesq
equation with self-consistent sources (BESCS-I) which reads as
\begin{eqnarray*}
  & -2  u_{2,x}- u_{1,xx}+ u_{1,\tau_2}  + \partial_x\sum_{i=1}^N  \phi_i \psi_i=0, \\
  & 3 u_{2 ,\tau_2}-3 u_{2 ,xx}+3 u_{1 ,x,\tau_2}  +6 u_1  u_{1 ,x}- u_{1 ,xxx}
  +3 \partial_x \sum_{i=1}^N  \phi_{i ,x} \psi_{i}=0, \\
  & \phi_{i ,xxx}  +3 u_1  \phi_{i ,x}
  +3 (u_{1 ,x}+u_2 ) \phi_i -\zeta^3 \phi_i=0, \\
  &  \psi_{i ,xxx}  +3 u_1  \psi_{i ,x}-3 u_2  \psi_i +\zeta^3 \psi_i=0,\qquad   i=1,\cdots,N.
\end{eqnarray*}
\end{eg}

\subsection{The $k$-constrained hierarchy of (\ref{eqns:nmcqKP})}
The $k$-constraint is given by \cite{MR1117170,MR1185854}
\begin{equation}
    L^k=B_k+\sum_{i=1}^N \phi_i\partial_q^{-1}\psi_i.
\end{equation}
By using the above $k$-constraint, it can be proved
that $L$ and $B_n$ are independent of $\tau_k$. By dropping $\tau_k$
dependency from (\ref{eqns:nmcqKP}), we get
 \begin{subequations}
  \label{eqs:k-constrained-q}
  \begin{eqnarray}
    &\left(B_k+\sum_{i=1}^N \phi_i\partial_q^{-1}\psi_i\right)_{t_n}
    =\left[(B_k+\sum_{i=1}^N \phi_i\partial_q^{-1}\psi_i)^{\frac n k}_+,
        B_k+\sum_{i=1}^N \phi_i\partial_q^{-1}\psi_i\right],\\
    &\phi_{i,t_n}=(B_k+\sum_{j=1}^N \phi_j\partial_q^{-1}\psi_j)^{\frac n k}_+(\phi_i),\\
    &\psi_{i,t_n}=-(B_k+\sum_{j=1}^N \phi_j\partial_q^{-1}\psi_j)^{\frac n k *}_+(\psi_i), \quad i=1,\cdots,N,
  \end{eqnarray}
\end{subequations}
which is the constrained $q$-deformed KP hierarchy. Some solutions
of the constrained $q$-deformed KP hierarchy can be represented by
$q$-deformed Wronskian determinant (see \cite{HLC} and the
references therein).

\parindent=0pt
{\bf Remark 4} In \cite{GGS2005,BSS2008}, the $k$-constrained $q$-KP
hierarchy can be constructed from
        the $q$-KP hierarchy by imposing the $k$-constraint. Here, the $k$-constrained $q$-KP hierarchy
        is obtained directly from the {\it extended} $q$-KP hierarchy (\ref{eqns:nmcqKP}) by dropping the $\tau_k$ dependence
        due to the $k$-constraint.

\begin{eg} [The second type of $q$-KdVSCS ($q$-KdVSCS-II)]
For $n=3$ and $k=2$, (\ref{eqs:k-constrained-q}) gives rise to the
second type of $q$-deformed KdV equation with self-consistent
sources ($q$-KdVSCS-II).
\begin{subequations}
\label{eqns:qKdVSCS-II}
\begin{eqnarray}
  & v_{1,t_3}+f_1-\sum_{i=1}^N h_{i1}=0, \\
  & v_{0,t_3}+f_0-\sum_{i=1}^N h_{i0}=0, \\
  & u_2+\theta(u_2) +\partial_q(u_1)+u_0 u_1
    +u_1 \theta^{-1}(u_0)   -\sum_{i=1}^N r_{i1}=0, \\
  & \phi_{i,t_3}=(\partial_q^3+s_2 \partial_q^2+s_1 \partial_q+s_0) (\phi_i), \\
  & \psi_{i,t_3}= - (\partial_q^3+s_2 \partial_q^2+s_1 \partial_q+s_0)^* (\psi_i),\qquad i=1,\cdots,N.
\end{eqnarray}
\end{subequations}
Let $q\rightarrow 1$ and $u_0\equiv 0$, then the $q$-KdVSCS-II
(\ref{eqns:qKdVSCS-II}) reduces to the second type of KdV equation
with self-consistent sources (KdVSCS-II or Yajima-Oikawa equation)
which reads as
\begin{eqnarray*}
  & u_2 = -\frac{1}{2} u_{1 ,x}+\frac{1}{2} \sum_{i=1}^N \phi_i  \psi_i, \\
  & u_{1 ,t_3}=\frac{1}{4}  u_{1 ,xxx}  +3 u_1   u_{1 ,x}
  +\frac{3}{4} \sum_{i=1}^N (\phi_{i ,xx} \psi_i -\phi_i  \psi_{i ,xx}), \\
  & \phi_{i ,t_3}= \phi_{i ,xxx}
  +3 u_1   \phi_{i ,x}+\frac{3}{2}  u_{1 ,x} \phi_i
  +\frac{3}{2} \phi_i \sum_{j=1}^N \phi_j  \psi_j, \\
  & \psi_{i ,t_3}= \psi_{i ,xxx}
  +3 u_1   \psi_{i ,x}+\frac{3}{2}  u_{1 ,x} \psi_i
  -\frac{3}{2} \psi_i \sum_{i=1}^N \phi_j  \psi_j,\qquad i=1,\cdots,N. \\
\end{eqnarray*}
\end{eg}

\begin{eg} [The second type of $q$-BESCS ($q$-BESCS-II)]
For $n=2$ and $k=3$, (\ref{eqs:k-constrained-q}) gives rise to the
second type of $q$-deformed Boussinesq equation with self-consistent
sources ($q$-BESCS-II))
\begin{subequations}
\label{eqns:qBESCS-II}
\begin{eqnarray}
    & s_{2,t_2}-f_2=0, \\
    & s_{1,t_2}-f_1 -\sum_{i=1}^N g_{i1}=0, \\
    & s_{0,t_2}-f_0 -\sum_{i=1}^N g_{i0} =0 , \\
    &  \phi_{i,t_2}=(\partial_q^2+v_1 \partial_q+v_0 ) (\phi_i), \\
    & \psi_{i,t_2}= -(\partial_q^2+v_1 \partial_q+v_0 )^* (\psi_i),\qquad i=1,\cdots,N.
\end{eqnarray}
\end{subequations}
Let $q\rightarrow 1$ and $u_0\equiv 0$, then the $q$-BESCS-II
(\ref{eqns:qBESCS-II}) reduces to the second type of Boussinesq
equation with self-consistent sources (BESCS-II) which reads as
\begin{eqnarray*}
    & -2  u_{2,x}- u_{1,xx}+ u_{1,t_2}=0, \\
    & 3 u_{2 ,t_2}-3 u_{2 ,xx}  +3  u_{1 ,x,t_2}+6 u_1    u_{1 ,x}
  - u_{1 ,xxx}-2 \partial_x\sum_{i=1}^N \phi_i  \psi_i =0, \\
    & \phi_{i,t_2}=\phi_{i,xx}  +2 u_1 \phi_i, \\
    & \psi_{i ,t_2}= - \psi_{i ,xx}  -2 u_1  \psi_i,\qquad i=1,\cdots,N.
\end{eqnarray*}
\end{eg}

\section{Conclusions}

A method is proposed in this paper to construct a new extended
$q$-deformed KP ($q$-KP) hiearchy and its Lax representation. This
new extended $q$-KP hierarchy contains two types of $q$-deformed KP
equation with self-consistent sources ($q$-KPSCS-I and
$q$-KPSCS-II), and its two kinds of reductions give the $q$-deformed
Gelfand-Dickey hierarchy with self-consistent sources and the
constrained $q$-deformed KP hierarchy. Thus, the reductions of the
new extended $q$-KP hierarchy may give some $q$-deformed
$(1+1)$-dimensional soliton equation with self-consistent sources,
e.g., the two types of $q$-deformed KdV equation with
self-consistent sources (including $q$-deformed Yajima-Oikawa
equation) and two types of $q$-deformed Boussinesq equation with
self-consistent sources. All of these results reduce to the
classical ones when $q \rightarrow 1$. The method proposed in this
paper is a general way to find $(1+1)$- and $(2+1)$-dimensional
$q$-deformed soliton equation with self-consistent sources and their
Lax representations.

\section*{Acknowledgments}
This work is supported by National Basic Research Program of China
(973 Program) (2007CB814800) and National Natural Science Foundation
of China (grand No. 10601028 and 10801083). RL Lin is supported in
part by Key Laboratory of Mathematics Mechanization.

\end{document}